\documentclass[11pt]{amsart}

\usepackage{booktabs}

\usepackage{amsmath,amsfonts,amsthm,mathrsfs,xspace,graphicx}
\usepackage[
colorlinks,bookmarks=true, citecolor=blue!90!black, urlcolor=blue!90!black,linkcolor=blue!90!black, pdfencoding=auto, psdextra]{hyperref}
\hypersetup{pdfauthor={Name}}
\usepackage{endnotes}
\usepackage{color}
\usepackage{float}
\usepackage{subcaption}
\usepackage{xprintlen}
\usepackage{array}
\newcolumntype{C}[1]{>{\centering\arraybackslash}p{#1}}
\newcolumntype{L}[1]{>{\raggedright\arraybackslash}p{#1}}

\usepackage{xcolor}
\definecolor{since}{rgb}{0.5,0.5,0.5}
\definecolor{newred}{HTML}{ff382e}
\definecolor{newgreen}{HTML}{549641}
\definecolor{newblue}{HTML}{4c4cfc}
\definecolor{neworange}{HTML}{c98702}

\usepackage{mdframed}
\usepackage{framed}
\usepackage{bbm}
\usepackage{suffix} 
\usepackage{tabularx}
\usepackage{makecell}
\usepackage{amssymb,latexsym}
\usepackage[capitalize]{cleveref}
\usepackage{enumitem}
\usepackage{multirow}
\usepackage{endnotes}
\usepackage{amssymb,latexsym}
\usepackage{enumitem}
\usepackage{float}
\usepackage{setspace}
\usepackage{soul}
\usepackage[section]{placeins}
\usepackage{thmtools}
\usepackage{thm-restate}
\usepackage{stmaryrd}
\usepackage{esvect}
\usepackage{mathtools, physics}
\usepackage{complexity}
\usepackage{verbatim}

\usepackage{algorithm}
\usepackage{algpseudocode}

\makeatletter
\renewcommand*\env@matrix[1][*\c@MaxMatrixCols c]{%
  \hskip -\arraycolsep
  \let\@ifnextchar\new@ifnextchar
  \array{#1}}
\makeatother

\newtheorem{theorem}{Theorem}[section] 
\newtheorem*{theorem*}{Theorem}
\newtheorem{proposition}[theorem]{Proposition}
\newtheorem*{proposition*}{Proposition}
\newtheorem{lemma}[theorem]{Lemma}
\newtheorem*{lemma*}{Lemma}
\newtheorem{corollary}[theorem]{Corollary}
\newtheorem*{corollary*}{Corollary}

\newtheorem*{conjecture*}{Conjecture}

\theoremstyle{definition}

\newtheorem*{fact*}{Fact}

\newtheorem{definition}[theorem]{Definition}
\newtheorem*{definition*}{Definition}

\DeclareRobustCommand{\EndDef}

\theoremstyle{remark}
\newtheorem{remark}[theorem]{Remark}
\AtEndEnvironment{remark}{\EndDef}
\AtEndEnvironment{definition}{\EndDef}
\newtheorem*{remark*}{Remark}

\DeclareMathOperator\im{im}
\DeclareRobustCommand{\euler}{\genfrac{\langle}{\rangle}{0pt}{}}

\DeclareMathOperator\Span{Span}
\DeclareMathOperator\wt{wt}

\DeclareMathOperator{\supp}{supp}

\DeclareMathOperator\dist{dist}

\usepackage{accents}

\newcommand\remove[1]{}{}

\newcommand{\br}[1]{\ensuremath{\left\{{#1}\right\}}}

\newclass{\QC}{QC}
\newcommand{\Coxeter}[1]{\ensuremath{\C_{W}({#1})}}

\newcommand{\coxeter}[3]{\ensuremath{\C_{{#1}}({#3})}}

\newcommand{\1}{{\mathbbm 1}}

\newcommand{\FF}{\ensuremath{\mathbb{F}}}

\newcommand{\ZZ}{\ensuremath{\mathbb{Z}}}

\newcommand{\mcB}{\ensuremath{\mathcal{B}}} 
\newcommand{\mcC}{\ensuremath{\mathcal{C}}}

 \newcommand{\bfs}{\ensuremath{\boldsymbol{s}}}
 \newcommand{\bft}{\ensuremath{\boldsymbol{t}}}

\DeclareMathAlphabet{\mathbbold}{U}{bbold}{m}{n}

\usepackage[margin=1.1in]{geometry}
\usepackage{enumitem,diagbox}
\usepackage{stfloats,nicefrac}

\newcommand{\parasum}[1]{b_{#1}}
\newcommand{\icomp}{{I^c}}
\newcommand{\rhosi}{\rho_{s,I}}
\newcommand{\parak}[1]{\mathcal{I}_{#1}}

\newcommand{\longest}[1]{w_0^{#1}}
\newcommand{\minrep}[1]{\mathrm{MinRep}({#1})}
\newcommand{\maxrep}[1]{\mathrm{MaxRep}({#1})}

\usepackage{cite}
\usepackage[normalem]{ulem}
\newcommand\redout{\bgroup\markoverwith{\textcolor{red}{\rule[0.5ex]{2pt}{0.8pt}}}\ULon}

\makeatletter
\newcommand\footnoteref[1]{\protected@xdef\@thefnmark{\ref{#1}}\@footnotemark}
\makeatother
\allowdisplaybreaks
\usepackage{balance}

\title[Distance and Decoding of Coxeter Codes]{Minimum Distance and Decoding of Coxeter Codes}

\author[A. Barg]{Alexander Barg}\address{University of Maryland, College Park, MD 20817, USA}\email{abarg@umd.edu} 
\author[Q. R. Gashi]{Q{\"e}ndrim R. Gashi}\address{University of Maryland, College Park, MD 20817, USA \& University of Prishtina, Kosovo}\email{qgashi@umd.edu}
\author[T. Xu]{Tianyuan Xu}\address{University of Richmond, Richmond, VA 23173, USA}\email{txu@richmond.edu}
\date{}

\begin{document}

\begin{abstract}
A binary Coxeter code associated with a finite Coxeter system $(W,S)$ is an $\FF_2$-linear span of indicators of standard cosets of a fixed rank. Coxeter codes, introduced in a recent paper by N. Coble and A. Barg,  are a generalization of Reed--Muller codes which arise when $W=\ZZ_2^m$ is the Coxeter group of type $mA_1$. 
In that paper, the authors proposed a conjectural value for the minimum distance of a general Coxeter code. This conjecture is proved in the present work. As a consequence, we obtain a Coxeter-theoretic generalization of Reed's majority-logic decoding algorithm for Reed--Muller codes.
\end{abstract}

\maketitle

\section{Introduction}
Let $(W,S)$ be a finite Coxeter system, where $W$ is a finite group with a generating set $S=\{s_1,\dots,s_m\}$ satisfying the defining relations $(s_is_j)^{M(i,j)}=1$ with $M(i,i)=1$ and $M(i,j)=M(j,i)\ge 2$ for $i\neq j$. A standard (left) coset of rank $r$ in $W$ is a coset of the form $wW_I$, where $w\in W, I \subseteq S$, and $W_I$ is a standard parabolic subgroup of rank $|I|=r$. We refer to \cite{BB05} for an introduction to the combinatorics of Coxeter groups. 

Recently, Coble and Barg \cite{coble2026coxeter} introduced a class of binary linear codes spanned by indicators of standard cosets of $W$ of a fixed rank. Specifically, they defined a {\em Coxeter code} $\Coxeter{r}$ of order $r\in\{-1,0,\dots,m\}$ to be the $\FF_2$-linear span of indicator vectors of the standard cosets of rank $m-r$:
   \begin{equation}\label{def: CC}
    \Coxeter{r} :=\Span\br{\1_{wW_I}\mid w\in W, I\subseteq S, \abs{I}=m-r}.
    \end{equation}
The value $-1$ is included for convenience since it enables one to properly formulate the duality results for Coxeter codes \cite[Sec.3]{coble2026coxeter}. At the same time, we have $\Coxeter{-1}=\{0\}$ by definition. The codes in ~\eqref{def: CC} form a direct generalization of a classic family of binary codes known as Reed--Muller codes $RM(r,m)$, which correspond to Coxeter type $mA_1$, where $S=\{e_1,\dots,e_m\}$ and $W=\ZZ_2^m$ is a direct product of $m$ permutation groups on 2 elements, with generators corresponding to the vectors of the standard basis. Reed--Muller codes have been extensively
studied in classical coding theory \cite[Ch.13--15]{MS77}, \cite{Assmus98}, \cite{abbe2023reed}, and they also give rise to a family of quantum CSS codes with a range of well-understood properties \cite{barg2024geometric}.

The Coxeter code $\Coxeter{r}$ has length $\abs{W}$ by 
construction, and its dimension is 
    \begin{equation}\label{eq:dimension}
        \dim \Coxeter{r} = \sum_{i=0}^r \euler{W}{i},
    \end{equation}
where the $W$-Eulerian number $\euler{W}{i}, i\in\{0,\dots,m\}$ is the count of elements in $W$ with descent number equal to $i$, \cite[Sec.7.2]{BB05}. For the classical Reed--Muller code $RM(r,m)$, this dimension reduces to the familiar expression $\sum_{i=0}^r \binom mi$; see
\cite{coble2026coxeter} for a proof for general $W$. 

In addition to length and dimension, the third important parameter of a code $C$ is its distance $\dist (C)$, which equals the minimum Hamming distance between two distinct codewords. For a linear code, one has $\dist (C)=\min_{x\in C\backslash\{0\}}\wt(x)$, where $\wt$ denotes the Hamming weight. The classical Reed--Muller codes are well known to have distance $\dist(RM(r,m))=2^{m-r}$ \cite[Thm.13.3]{MS77}. For a general Coxeter code $\Coxeter{r}$, it is clear from \eqref{def: CC} that $\dist(\Coxeter{r})$ is at most the size of the smallest standard parabolic subgroup of rank $m-r$. 
Coble and Barg conjectured that this upper bound is in fact the exact distance for general Coxeter codes.
In this paper, we prove this conjecture:

\begin{theorem}\label{thm: main}
 Let $(W,S)$ be a finite Coxeter system of rank $m$ and let $d_r:=\min_{I\subseteq S, \abs{I}=m-r} |W_I|$
 be the order of the smallest parabolic subgroup of rank $m-r$.
    Then
       $$
       \dist(\Coxeter{r})=d_r.
       $$
\end{theorem} 

Our proof of Theorem \ref{thm: main} proceeds by reduction to certain projections of the code $\Coxeter{r}$ which we call {\em shadow codes}, and which we explore on their own merit, finding their bases and parameters. The core part of the proof is a lower bound on the distance of a shadow code, which is then taken back to the main code by peeling off fixed-rank layers starting from the top level. While our arguments are phrased in combinatorial terms, the key ideas of the proof can be adapted to produce a decoding algorithm for Coxeter codes that generalizes Reed's majority-logic decoding algorithm for Reed--Muller codes  \cite{reed1954codes}. As one of our results, we uncover this connection in Sec.~\ref{sec:decoding}. Our decoding algorithm 
presented in Sec.~\ref{sec: decoder} also yields another proof of Theorem~\ref{thm: main}.


\section{Preliminaries}
\label{sec:notation} {All codes that we mention are binary linear codes, all group actions are left actions, and all modules are left modules.} 
For the rest of the paper, let $(W, S)$ be a finite Coxeter system with rank
$\abs{S}=m$. Let $\ell$ and $\le$ be the length function and Bruhat order on $W$,
respectively. We recall that a \emph{right descent} of an element $w\in W$ is an
element $s\in S$ such that $ws<w$ (equivalently, such that $\ell(ws)<\ell(w)$),
and we write $D_R(w)$ for the set of all right descents of $w$. For every
$K\subseteq S$, we denote the set $S\setminus K$ by $K^c$ and the standard parabolic subgroup $\langle K\rangle$ of
$W$ by $W_K$, we define 
\[ 
  W^K=\{w\in W:D_R(w)\cap K=\varnothing\},
\]
and we set
\[\euler{W}{K}=\{w\in W: D_R(w)=K\}.\]
The numbers $|\euler{W}{K}|$ refine the \emph{Eulerian numbers} $\euler{W}{i}$, defined as
\[
\euler{W}{i}=\abs{\{w\in W:\abs{D_R(w)}=i\}},\]
in that $\euler{W}{i}=\sum_{K\subseteq S, \abs{K}=i}|{\euler{W}{K}}|$.

Let  
\[
w\star s:=
\begin{cases}
ws,&\text{if } ws>w,\\
w,&\text{if } ws<w
\end{cases}
\]
for any $w\in W$ and $s\in S$. For any word $\bfs=s_1\cdots s_m$, the \emph{Demazure product} \cite{demazure} of $\bfs$ is the recursively defined  element
\[
\delta(\bfs)=(((e\star s_1)\star s_2)\cdots)\star s_m \in W.
\]
Note that if $w\in W$ is the element expressed by $\bfs$, then $\delta(\bfs)=w$ if and only if $\bfs$ is a reduced word of $w$.

We will frequently use the following standard facts in this paper.
\begin{itemize}
    \item[(C1)] \cite[Sec.2.2]{BB05}
    \makeatletter\def\@currentlabel{C1}\makeatother\label{cond: subword}
    If $w,u\in W$ and $\bfs=s_1\cdots s_m$ is a reduced word of $w$, then $u\le w$ if and only if some subword of $\bfs$ is a reduced word of $u$, where by a subword of $\bfs$ we mean a word of the form $s_{i_1}\cdots s_{i_k}$ for some $1\le i_1<\cdots<i_k \le m$.
    \item[(C2)] \cite[Sec.2.3]{jahnstump} (cf. \cite[Lemma 3.4]{knutson2004subword})
    \makeatletter\def\@currentlabel{C2}\makeatother\label{cond: demazure}
For any word $\bfs=s_1\cdots s_m$, the Demazure product
$\delta({\bfs})$ is the unique Bruhat-maximal element of $W$ that can be
obtained as the product of a subword of $\bfs$. In particular, $\bfs$ contains a reduced word of $\delta(\bfs)$ as a subword.
  \item[(C3)] \cite[Sec.2.3]{BB05} \makeatletter\def\@currentlabel{C3}\makeatother\label{cond: longest} 
  A finite Coxeter system $(W',S')$ has 
    a unique maximal element, $w_0$, with respect to the Bruhat order. 
    The element $w_0$ has order 2, and it is the unique element of $W'$ such that
    $D_R(w_0)=S'$. We also have $D_R(w_0w)=S'\setminus D_R(w)$ for all $w\in W'$.
  \item[(C4)] \cite[Sec.2.4]{BB05}\makeatletter\def\@currentlabel{C4}\makeatother\label{cond: intersect} 
  For any $I\subseteq S$, the pair $(W_I,I)$ is
    also a Coxeter system, and for every $w\in W_I$ we have $\ell_I(w)=\ell(w)$,
    where $\ell_I$ stands for the length function associated with $(W_I,I)$. For
    any $I, J\subseteq S$, we have $W_I\cap W_J=W_{I\cap J}$.
\item[(C5)] \cite[Sec.2.4]{BB05}
\makeatletter\def\@currentlabel{C5}\makeatother\label{cond: factor}
For every $w\in W$ and $I\subseteq S$, there is
  a unique factorization $w=w^I w_I$, called the \emph{(left) coset
    factorization
  of $w$ with respect to $I$}, such that $w^I\in W^I$ and $w_I\in
  W_I$, and for this factorization we have $\ell(w)=\ell(w^I)+\ell(w_I)$.
\end{itemize}

For each $I\subseteq S$, we will denote the longest element of the parabolic
subgroup $W_I\le W$ by
$\longest{I}$. The following well-known facts follow
readily from \eqref{cond: longest}--\eqref{cond: factor} and will be particularly useful in Sec. \ref{sec:shadow}. 
\begin{itemize}
  \item[(C6)] \makeatletter\def\@currentlabel{C6}\makeatother\label{cond: coset_rep}  For any $I\subseteq S$ and any $w=w^I\cdot w_I\in W$, the element
    $w^I$ is the unique minimal-length representative of the left
    coset $wW_I$, and the element $w^I\longest{I}$ is the unique maximal-length coset
  representative of $wW_I$.  
The set of all minimal-length representatives of
the left cosets of $W_I$, $\minrep{I}$, coincides with $W^I$; the set of all
maximal-length representatives of
the left cosets of $W_I$, $\maxrep{I}$, is given by  
  \begin{equation}
    \label{eq:maxrep1}
    \maxrep{I}= \{w\in W: w_I = \longest{I}\}=\{w\in
    W: D_R(w_I)=I\}. 
  \end{equation}
  Furthermore, we have $D_R(w_I)= D_R(w)\cap I$ for any $w\in W$, so we also have
\begin{equation}
\label{eq:maxrep2}
\maxrep{I}= \{w\in W: I\subseteq D_R(w)\}.
\end{equation}
\end{itemize}
\noindent
Let $R=\FF_2[W]$ be the
group algebra of $W$ over $\FF_2$. Every element of $R$ has the form $t_X=\sum_{w\in X}w$ for some $X\subseteq W$, in which case we will call $X$ the \emph{support} of $t_X$ and identify $t_X$ with the indicator function $\1_X$ on $W$. 
We write 
    \[
\parasum{I}:=t_{W_I}=\sum_{w\in W_I}w\in R 
   \] 
for each $I\subseteq S$, with $W_\varnothing=\{1_W\}$.
It
follows from \eqref{def: CC} that we can view Coxeter codes as sums of suitable cyclic $R$-submodules of the form $Rb_I, I\subseteq S$. 
More precisely, for each $r\in \{-1, 0, \cdots, m\},$ we have 
  \[
  \Coxeter{r}=\sum_{I\subseteq S, \abs{I}=m-r} R\,\parasum{I}.
  \]
Elements of the form $wb_I$ will play an important role in this paper. 

\begin{lemma}\label{lem: intersection}\phantom{lnebreak}\\[-.2in]
\begin{enumerate}
\item[{\rm(a)}] We have $wb_I=w'b_I$ if and only if $wW_I=w'W_I$.
\item[{\rm(b)}] If $I\cap J\ne\varnothing$, then $\parasum{I}\parasum{J}=0$ in $\FF_2[W]$.
\end{enumerate}
\end{lemma}
\begin{proof}
Part (a) holds because the support of $xb_I$ equals the coset $xW_I$ for all $x\in W$. Part (b) follows from \eqref{cond: intersect} as follows.
We have 
  \[
    b_Ib_J= \sum_{u\in W_I}u \sum_{v\in W_J} v = \sum_{x\in W_IW_J} a_x x,
  \]
  where  $a_x$ is the number of pairs $(u,v)$ where $u\in W_I, v\in W_J$, and
  $uv=x$ for every $x\in W_IW_J$. 
  It is a basic fact from group theory that this
  number is precisely the order of the group
  ${W_I\cap W_J}$. We have $W_I\cap W_J=W_{I\cap J}$ by \eqref{cond: intersect}, which has even order if $I\cap J\neq \varnothing$, so $a_x=0$ for all
  $x\in W$ and $b_Ib_J=0$ whenever $I\cap J\neq \varnothing$. 
\end{proof}

  We recall from \cite[Theorem 3.6, Theorem 4.1]{coble2026coxeter} that the set 
\[\mathcal{B}_{m-r} := \{wb_J: J=S\setminus D_R(w) \text{\ and\ } \abs{J}\ge m-r\}\]
forms a basis of $\Coxeter{r}$, which also implies the expression for $\dim\Coxeter{r}$ in \eqref{eq:dimension}.

We will be using the following related but different basis of $\Coxeter{r}$ in this paper. 
\begin{proposition}
\label{prop:basis}
For every $r\in \{0,\dots, m\}$, 
the set 
\begin{align*}
\mcC_{m-r}&:=\{wb_J: J=D_R(w), |J|\ge m-r\}
\end{align*}
is a basis of $\Coxeter{r}$. 

\end{proposition}
\begin{proof}
    This fact has been mentioned without proof in \cite[Sec.3.1]{coble2026coxeter}, so we provide a short proof here. Consider the action of the longest element $w_0\in W$ on $\Coxeter{r}$. Recall that $w_0^2=1_W$ and $D_R(w_0w)=S\setminus D_R(w)$ for all $w\in W$ by \eqref{cond: longest}, so this action is an involutive linear map on $\Coxeter{r}$ that interchanges $\mathcal{B}_{m-r}$ and $\mathcal{C}_{m-r}$. Since $\mathcal{B}_{m-r}$ is a basis, it follows that so is $\mathcal{C}_{m-r}$.
\end{proof}

\begin{remark}
\label{rmk:symm}
The fact that $D_R(w_0w)=S\setminus D_R(w)$ for all $w\in W$ also implies that left multiplication by $w_0$ on $W$ interchanges the sets $\euler{W}{J}$ and $\euler{W}{J^c}$ for any $J$, which implies that 
\begin{equation}
    \label{eq:ds}
    \euler{W}{i}=\euler{W}{m-i}
\end{equation} 
for any $0\le i\le m$. Eq.~\eqref{eq:ds} is known as the \emph{Dehn--Sommerville relation}.
\end{remark}


\section{Shadow codes}
\label{sec:shadow}
For $K\subseteq S$, let
          \[
          M_K=\FF_2[W/W_K].
          \] 
Define the \emph{$K$-shadow projection} to be the $\FF_2$-linear map given by  
  \[
    \pi_K: R\rightarrow M_K, \quad w\mapsto  wW_K, 
  \]
and define the  \emph{$K$-lift} to be the  $\FF_2$-linear map given by 
  \[
  \iota_K:  M_K\rightarrow R, \quad wW_K\mapsto wb_K. 
\]
Note that $\iota_K$ is well defined by Lemma \ref{lem: intersection}(a). Note also
that $M_K$ has an $R$-module structure induced by the natural action of $W$ on
the left cosets of $W_K$, and that the shadow projection $\pi_{K^c}$ is
$W$-equivariant and is thus an $R$-module homomorphism. 

\begin{definition}
  \label{def:shadow}
For every $I\subseteq S$, we define the \emph{shadow code of type $I$} to be the image
\[
D_I=\pi_{\icomp}(R\parasum{I}),
\]
viewed as a linear code in the vector space $M_{\icomp}$, and we set 
 \[
 \mcB_I:=\{\pi_{\icomp}(wb_I): D_R(w)=I\}.
 \]
For any $r\in \{-1, 0,\dots, m\}$, denote
   $$
\parak{m-r}:=\{I\subseteq S: \abs{I}=m-r\}.
   $$
   Further, let
$\phi_r$ be the map
\begin{equation}\label{eq: phi map}
      \phi_r:\Coxeter{r}\to\bigoplus_{I\in \parak{m-r}} D_I, \quad 
    c\mapsto \left(\pi_{\icomp}(c)\right)_{I}.
\end{equation}
\end{definition}

The main goal of this section is to prove that for every $0\le r\le m$, the natural embedding
$\Coxeter{r-1}\hookrightarrow\Coxeter{r}$ followed by $\phi_r$ forms a short exact sequence. This is the
content of Theorem \ref{thm:ses}.  We will also show that $\mcB_I$ is a basis of the code $D_I$ for all $I\subseteq S$. 

\begin{proposition}
\label{prop:pi}
Let $K\subseteq S$ and $x\in R$. Let $r=\abs{K}$.
  \begin{enumerate}
        \item[\upshape(a)] If $x=\sum_{w\in W} a_ww\in R$, then
  \[
\pi_K(x)=\sum_{\Omega\in W/W_K}\Big(\sum_{w\in \Omega}a_w\Big) \Omega.
  \]
\item[\upshape(b)] The $K$-shadow projection weakly decreases weight: $\wt(\pi_{K}(x))\le \wt(x)$.
    \item[\upshape(c)] We have $\iota_K\pi_K(x)=xb_K$, and $\pi_K(x)=0$ if and only if
      $xb_K=0$.  
\item[\upshape(d)] For any $w\in W$ and any $J\subseteq S$ such that $J\cap K\neq
  \varnothing$, we have $\pi_K(wb_J)=0$.
  \item[\upshape(e)] We have $\pi_{K}(\Coxeter{r-1})=0$.
    \end{enumerate}
\end{proposition}

\begin{proof}
If $x=\sum_{w\in W}a_ww$, then we have $\pi_K(x)=\sum_{w\in W}a_wwW_K$, so each coset $\Omega=uW_K$ of $W_K$ appears with coefficient 
\[\sum_{w\in W: wW_K=uW_K}a_w=\sum_{w\in \Omega}a_w.\]
Part (a) follows. It also follows that every coset in the support of $\pi_K(x)$ arises from at least one element in the support of $x$, so (b) holds as well.

  We have $\iota_K\pi_K(w)=wb_K$ for any $w\in W$ by the definition of
  $\iota_K$ and $\pi_K$, so it follows 
  from linearity that $\iota_K\pi_K(x)=xb_K$ for any $x\in R$. 

  Since distinct left cosets of $W_K$ have disjoint support, the map $\iota_K$
  is injective, so we have $\pi_K(x)=0$ if and only if $\iota_K\pi_K(x)=0$,
  which occurs if and only if $xb_K=0$ by the preceding statement. This completes the
  proof of (c).  

  If $J\cap K\neq \varnothing$, then $(wb_J)b_K=w(b_Jb_K)=0$ by Lemma \ref{lem: intersection}(b), so we have $\pi_K(wb_J)=0$ by (c). This proves (d).
   
  By Proposition \ref{prop:basis}, the set $\mathcal{C}_{m-(r-1)}=\{wb_J:J=D_R(w), \abs{J}\ge m-r+1\}$ is a basis of $\Coxeter{r-1}$. For every set $J\subseteq S$ with $\abs{J}\ge m-r+1$, we have $\abs{K}+\abs{J}=m+1>\abs{S}$ and thus $K\cap J\neq \varnothing$, so every element of $\mathcal{C}_{m-(r-1)}$ vanishes under $\pi_K$ by (d). It follows that 
  $\pi_{K}(\Coxeter{r-1})=0$, which proves (e).
\end{proof}

\begin{corollary}
     \label{cor:surj}
    The map $\phi_r$ is surjective.
\end{corollary}

\begin{proof}
Let $z=(z_I)_{I}\in \bigoplus_{I\in \parak{m-r}}D_I$. By the
definition of $D_I$, there exists $c_I\in Rb_I$ such that $z_I=\pi_{\icomp}(c_I)$
for each $I\in \parak{m-r}$. Set $c=\sum_{J\in \parak{m-r}} c_J$. For any two
distinct sets $I,J\in \parak{m-r}$, we have $\icomp\cap J\neq\varnothing$, so Proposition \ref{prop:pi}(d) implies  $\pi_{\icomp}(c_J)=0$. This
further implies that the $I$-component of $\phi_r(c)$ is 
    \[
      \phi_I(c):=\pi_{\icomp}\Big(\sum_{J\in \parak{m-r}} c_J\Big)=\pi_{\icomp}(c_I)=z_I
    \]
    for every $I\in \parak{m-r}$, so $\phi_r(c)=z$. It follows that $\phi_r$ is
    surjective, as desired.
\end{proof}

The next result proves the linear independence of the set $\mcB_I$. We will use it to help prove 
$\ker\phi_r=\Coxeter{r-1}$, which we will then use to help prove that $\mcB_I$ is in fact a basis of $D_I$.
\begin{proposition}
\label{prop:independence}
 For any $I\subseteq S$, the map
 \[f_I: \euler{W}{I} \to \mcB_I, \quad w\mapsto\pi_{\icomp}(wb_I)\] is a bijection. 
 The elements of the set $\mcB_I$ are linearly independent over $\FF_2$.
\end{proposition}

\begin{proof}
Let $w\in \euler{W}{I}$, so that $D_R(w)=I$.
Proposition \ref{prop:pi}(c) implies that 
\[
 \iota_{\icomp}(f_I(w))= \iota_{\icomp}\pi_{\icomp}(wb_I)=wb_Ib_{\icomp}=\sum_{u\in wW_I, y\in
  W_{\icomp}}uy\in R.
\]
Set $z=w_0^{\icomp}$, the unique longest element in the parabolic subgroup $W_{\icomp}$.
Since $D_R(w)=I$, it follows from \eqref{cond: factor} that the coset factorization of the element $wz$ with respect to $\icomp$ is simply
$wz=w\cdot z$, with $\ell(wz)=\ell(w)+\ell(z)$, and it follows from \eqref{cond: coset_rep} that $w$ is the unique maximal-length element in the coset $wW_I$. Thus, for any pair of
elements $u\in wW_I$ and $y\in W_{\icomp}$, we have 
\[ \ell(uy)\le \ell(u)+\ell(y)\le \ell(w)+\ell(z)=\ell(wz), 
\] 
where $\ell(uy)=\ell(wz)$ if and only if $u=w$ and
$y=z$. This implies that the element $w':=wz$ is
the unique longest element in the support of $\iota_{\icomp}(f_I(w))$.
Thus, we may recover $w$ from $f_I(w)$ by taking the
$\icomp$-lift, finding the unique longest element $w'$ in the support of
the $\icomp$-lift, and then computing $w$ as $w=w'z^{-1}$. This implies $f_I$ is injective.

We also have $\im f_I=\mcB_I$ by the definition of $f_I$ and $\mcB_I$, so $f_I$ is a bijection. 

It remains to show that $\mcB_I=\im f_I$ is linearly independent. Suppose $\sum_{j=1}^k
a_j f_I(w_j)=0$ for some scalars $a_1, \dots, a_k\in \FF_2$ and distinct elements $w_1, ..., w_k\in \euler{W}{I}$. 
By the last paragraph, if $w_i$ is any element in
$w_1, \dots, w_k$ with maximal length, then
in the $\icomp$-lift
\[
  \iota_{\icomp} \Big(\sum_{j=1}^k a_j f_I(w_j)\Big)
  =\sum_{j=1}^ka_j \left(\iota_{\icomp} f_I(w_j)\right)\in R
\]
the element $w_iz$ is a maximal-length element that appears with coefficient $a_i$, so $a_i=0$. It follows
by induction on $k$ that $a_j=0$ for all $1\le j\le k$, so $\mathcal{B}_I$ is
linearly independent, as desired.
\end{proof}

It is important to note that Theorem \ref{thm:ses} below allows us to prove the main distance result using induction. Indeed, the fact that there is a short exact sequence as in the theorem means that a codeword $c\in C_W(r)$ either lies in $C_W(r-1)$ or projects non-trivially onto the direct sum, where we will be able to have precise control over the distances of the shadow codes $D_I$ (Proposition \ref{prop:14}). We will use the distances of these shadow codes and the fact that shadow projections do not increase weight to prove Theorem \ref{thm: main}.

\begin{theorem}
\label{thm:ses}
For every $r\in \{0,\dots, m\}$, we have a short exact sequence of $R$-modules given by
\begin{equation}\label{eq:ses}
  0\longrightarrow \Coxeter{r-1}\xhookrightarrow{\psi_{r-1}}
  \Coxeter{r}\xlongrightarrow{\phi_r}
  \bigoplus_{I\in \parak{m-r}} D_I\longrightarrow 0,
\end{equation}
where $\psi_{r-1}$ is the natural embedding of $\Coxeter{r-1}$ into $\Coxeter{r}$.
\end{theorem}

\begin{proof}
The map $\psi_{r-1}$ is an injective $R$-module homomorphism because it is a natural embedding. The map $\phi_r$ is an $R$-module homomorphism because we have noted that $\phi_I$ is an $R$-module homomorphism for every $I\subseteq S$, and $\phi_r$ is surjective by Corollary \ref{cor:surj}. 
It remains to show that 
$\im \psi_{r-1}=\ker \phi_{r}$. In other words, it suffices to show
    that $\Coxeter{r-1}=\ker \phi_r$.  
    Proposition \ref{prop:pi}(d) implies that $\pi_{\icomp}(\Coxeter{r-1})=0$
    for all $I\in \parak{m-r}$, so we have $\Coxeter{r-1}\subseteq\ker \phi_r$
    and it further suffices to show that $\ker\phi_r\subseteq \Coxeter{r-1}$. 

Let $c\in \ker \phi_r\subseteq \Coxeter{r}$. Expand $c$ into the basis $\mcC_{m-r}$ of $\Coxeter{r}$ as
  \[
  c=\sum_{w:\abs{D_R(w)}\ge m-r} a_w e_w,
  \]
where $e_w=wb_{D_R(w)}$ for every $w$. If $\abs{D_R(w)}>m-r$, then $e_w$ lies in the basis $\mcC_{m-(r-1)}$ of $\Coxeter{r-1}$, so $\phi_r(e_w)=0$ since $\Coxeter{r-1}\subseteq \ker \phi_r$.
It follows that
\[
0=\phi_r(c)=\sum_{w:\abs{D_R(w)=m-r}}\phi_r(a_we_w).\]
Now fix a set $I\subseteq S$ with $\abs{I}=m-r$. Taking the $I$-component of the above sum yields
\[
0=\sum_{w:\abs{D_R(w)}=m-r}a_w\pi_{\icomp}(e_w)=\sum_{J\in \parak{m-r}}\sum_{w\in \euler{W}{J}}a_w\pi_{\icomp}(e_w).
\]
For any $J\in \parak{m-r}\setminus\{I\}$, we have $\icomp\cap J\neq \varnothing$, and therefore $\pi_{\icomp}(e_w)=0$ for all $w\in \euler{W}{J}$ by Proposition \ref{prop:pi}(d). It then follows that 
\[
\sum_{w:w\in \euler{W}{I}}a_w\pi_{\icomp}(e_w)=0. 
\]
The set $\{\pi_{\icomp}e_w:w\in \euler{W}{I}\}$ is precisely the set $\mcB_I$, so 
Proposition \ref{prop:independence} implies that $a_w=0$ for all $w\in \euler{W}{I}$. Since $I$ is arbitrary, it follows that $a_w=0$ for all $w\in W$ with $\abs{D_R(w)}=m-r$, so $c$ is in the span of the basis $\mcC_{m-(r-1)}$ of $\Coxeter{r-1}$. This proves $\ker \phi_r\subseteq\Coxeter{r-1}$, as desired.
\end{proof}
The next theorem is not needed to show our main result and is included to provide an additional insight into the properties of shadow codes.
\begin{theorem}
\label{thm:shadow_basis}
For any $I\subseteq S$,
the set $\mathcal{B}_I=\{\pi_{\icomp}(wb_I): D_R(w)=I\}$
is a basis of the shadow code $D_I$, and we have \[\dim D_I=\abs{\euler{W}{I}}.\]
\end{theorem}

\begin{proof}
Let $I\subseteq S$ and suppose $\abs{I}=m-r$ for some $0\le r\le m$. 
Proposition \ref{prop:independence} implies that $\abs{\mcB_I}=\abs{\euler{W}{I}}$, so it suffices to prove that $\mcB_I$ is a basis of $D_I$.

For all $J\in \parak{m-r}$, the set $\mathcal{B}_J$ is linearly independent by Proposition \ref{prop:independence}, so we have $\dim D_J\ge \abs{B_J}$, where equality holds if and only if $B_J$ is a basis for $D_J$. Since $\abs{B_J}=\euler{W}{J}$ by Proposition \ref{prop:independence}, it follows that 
\[
\dim\bigoplus_{J\in \parak{m-r}}D_J=\sum_{J\in \parak{m-r}}\dim D_J\ge \sum_{J\in \parak{m-r}} \abs{\mathcal{B}_J}=\sum_{J\in \parak{m-r}}\euler{W}{J}=\euler{W}{m-r},
\]
with $\dim \bigoplus_{J\in \parak{m-r}}D_J=\euler{W}{m-r}$ if and only if $\mathcal{B}_J$ is a basis of $D_J$ for all $J\in \parak{m-r}$. Thus, to prove the theorem it further suffices to show that $\dim \bigoplus_{J\in \parak{m-r}}D_J=\euler{W}{m-r}$. This follows from Theorem \ref{thm:ses} as follows: since the map $\phi_r: \Coxeter{r}\rightarrow\bigoplus_{I\in \parak{m-r}}D_I$ is surjective and has $\Coxeter{r-1}$ as its kernel by Theorem \ref{thm:ses}, we have 
   \begin{align*}
\dim\bigoplus_{J\in \parak{m-r}}D_J&=\dim\Coxeter{r}-\dim\Coxeter{r-1}\\[-.1in]
&=\sum_{i=0}^{r}\euler{W}{i}-\sum_{i=0}^{r-1}\euler{W}{i}\\
&=\euler{W}{r}=\euler{W}{m-r},
  \end{align*}
where the second and fourth equalities hold by Eqns.~\eqref{eq:dimension} and \eqref{eq:ds}, respectively.
\end{proof}


\section{Minimum distance of Coxeter codes}
In this section, we prove Theorem~\ref{thm: main}. 
 We will first show that $\dist(D_I) = \abs{W_I}$ for any $I\subseteq S$ (Proposition \ref{prop:14}), which we will then combine with Theorem \ref{thm:ses} to deduce Theorem \ref{thm: main}. 

\begin{definition}
\label{def:apartment}\phantom{Definition}\\[-.1in]
    \begin{enumerate}
        \item For any $I\subseteq S$ and $s\in I$, we define the \emph{$s$-coarsening of $M_{I^c}$} to be the  $\FF_2$-linear map 
\[\rhosi: M_{I^c}\rightarrow M_{I^c\cup \{s\}}, \quad 
 wW_{\icomp}\mapsto wW_{\icomp\cup \{s\}},
 \]
and we write $\rhosi$ as $\rho_s$ if $I$ is clear from context.
\item 
For each $x\in W$, we define the \emph{$I$-apartment indexed by $x$} to be the set
\[
\Sigma_x:=\{\,xu\,W_{\icomp} : u\in W_I\,\}\subseteq W/W_{\icomp}.
\]
We call each element in $\Sigma_x$ an \emph{$I$-chamber}.
    \end{enumerate}
\end{definition}
 
Note that $\rhosi$ is well defined because $W_{\icomp} \le W_{\icomp\cup \{s\}}$. The map sending each $u\in W_I$ to the coset $xuW_{\icomp}$ is injective, because if $xuW_{\icomp}=xu'W_{\icomp}$ for $u,u'\in W_I$ then we have
$u^{-1}u'\in W_{\icomp}\cap W_I=W_{I\cap {\icomp}}=\{1_W\}$, so each $I$-apartment $\Sigma_x$ consists of $\abs{W_I}$ distinct chambers.

\begin{lemma}
    \label{lem:shadow}
If $s\in I\in \parak{m-r}$ for some $0\le r\le m$, then we have $\rho_s(D_I)=0$. 
\end{lemma}

\begin{proof}
    It suffices to show that $\rho_s(\pi_{\icomp}(w\parasum{I}))=0$ for any $w\in W$.  We have
\begin{equation}
\label{eq:rhopi}
\rho_s\bigl(\pi_{\icomp}(w\parasum{I})\bigr)=\sum_{u\in W_I} wu\,W_{{\icomp}\cup\{s\}} .
\end{equation}
For any $u_1,u_2\in W_I$, we have $wu_1W_{{\icomp}\cup \{s\}}=wu_2W_{{\icomp}\cup \{s\}}$ if and only if 
\[u_1^{-1}u_2\in W_I\cap W_{{\icomp}\cup\{s\}}=W_{I\cap({\icomp}\cup\{s\})}=W_{\{s\}},
\]
where the first set equality holds by \eqref{cond: intersect}. It follows that the fibers of the
map $W_I\to W/W_{{\icomp}\cup\{s\}}$, $u\mapsto wu\,W_{{\icomp}\cup\{s\}}$ are simply the cosets of $W_{\{s\}}$ in $W_I$. All of these cosets have size $\abs{W_{\{s\}}}=\abs{\{1_W,s\}}=2$, so Eq.~\eqref{eq:rhopi} implies that $\rho_s(\pi_{\icomp}(wb_I))=0$, as desired.
\end{proof}

\begin{lemma}[Parabolic Bruhat projection]\label{lem:B}
For each $w\in W$ the set $\{u\in W_I:u\le w\}$ has a unique maximum element, denoted
$\beta_I(w)\in W_I$, with respect to the Bruhat order. Moreover, for every $s\in I$, every $y\in W^{{\icomp}\cup\{s\}}$, and
every $q\in W_{{\icomp}\cup\{s\}}$, we have
\[
\beta_I(yq)\in\beta_I(y)\,W_{\{s\}}=\{\beta_I(y), \beta_I(y)s\}.
\]
\end{lemma}

\begin{proof}
Fix a reduced word $\bfs=s_1\cdots s_m$ of $w$, and let $\bft$ be the subword of $\bfs$ obtained by removing all letters $s_i$ not in $I$. Put $b:=\delta({t})\in W_I$. The word $\bft$ contains a reduced word of $b$ as a subword by \eqref{cond: demazure}, and hence so does $\bfs$, which implies that $b\le w$ by \eqref{cond: subword}.
Conversely, if $u\in W_I$ and $u\le w$, then \eqref{cond: subword} implies that some
subword of $s_1\cdots s_m$ is a reduced word for $u$. Since $u\in W_I$, this subword
uses only letters from $I$, so it is a subword of ${t}$ and $u\le\delta({t})=b$ by \eqref{cond: demazure}.
It follows that $b$ is the unique maximum element of the set $\{u\in W_I: u\le w\}$ (and is independent of the choice of the reduced word $\bft$); in other words, we have $\beta_I(w)=\delta({\bft})$.

For the second assertion, let $s\in I, y\in W^{{\icomp}\cup\{s\}}$ and $q\in W_{\icomp\cup \{s\}}$. By \eqref{cond: factor}, we have $\ell(yq)=\ell(y)+\ell(q)$, so concatenating a reduced word $\bfs$ of $y$ with a reduced word $\bft$ of $q$ produces a reduced word for
$yq$. Since $q\in W_{\icomp\cup\{s\}}$, the only letter from $I$ that may appear in $\bft$ is $s$, so it follows from the last paragraph and the definition of $\star$ that
\[
\beta_I(yq)=\delta(\bfs\bft)=(((\delta(\bfs)\star s)\star) \cdots \star s)=(((\beta_I(y)\star s)\star s)\cdots\star s) \in \beta_I(y)W_{\{s\}}.
\]
This completes the proof.
\end{proof}

\begin{lemma}[Rank-selected retraction]\label{lem:C}
Fix $x\in W^{\icomp}$ and let $C:=xW_{\icomp}\in\Sigma_x$. {Consider the map $r_x:W/W_{\icomp}\to\Sigma_x$ given by $\Omega\mapsto  x r_0(x^{-1}\Omega)$, where $r_0(vW_{I^c})=\beta_I(v)W_{I^c}$ for any minimum-length coset representative $v\in W^{\icomp}$ (and $\beta_I(v)$ is as defined in Lemma \ref{lem:B}).} Then
\begin{itemize}
\item[{\rm (a)}] $r_x$ fixes the apartment $\Sigma_x$ pointwise;
\item[{\rm (b)}] for each $s\in I$, $r_x$ carries every fiber of $\rho_s$ into a single
fiber of $\rho_s$; consequently, there is a map $r^{(s)}:W/W_{\icomp\cup\{s\}}\rightarrow W/W_{\icomp\cup\{s\}}$ such that
$\rho_s\circ r_x=r^{(s)}\circ\rho_s$;
\item[{\rm (c)}] $r_x^{-1}(C)=\{C\}$.
\end{itemize}
\end{lemma}

\begin{proof}
We first treat the case $x=1_W$, where $C=C_0:=W_{\icomp}$, $\Sigma_x=\Sigma_0:=\{uW_{\icomp}:u\in W_I\}$, and $r_x=r_0$. 
Let us prove that the mapping $r_0$ satisfies properties (a)--(c).

(a) Each chamber in $\Sigma$ takes the form $\Omega=uW_{\icomp}\in \Sigma_0$ for some $u\in W_I$. We have $D_R(u)\subseteq I$, so $u\in W^{\icomp}$ and $u$ is the minimal representative of $\Omega$ by \eqref{cond: coset_rep}. 
Since $u\in W_I$, we also have $\beta_I(u)=u$ by Lemma \ref{lem:B}, so $r_0(\Omega)=r_0(uW_{\icomp})=\beta_I(u)W_{\icomp}=uW_{\icomp}=\Omega$.

{(b)} Fix $s\in I$ and put $J:={\icomp}\cup\{s\}$. If $wW_{\icomp}$ is a coset in a fiber $\rho_s^{-1}(yW_J)$, where $w\in W^{\icomp}$ and $y\in W^J$ are the respective minimal representatives of $wW_{\icomp}$ and $yW_J$, then we have $w=yq$ for some $q\in W_J$ such that $\ell(w)=\ell(y)+\ell(q)$ by \eqref{cond: factor}, and \eqref{cond: intersect} then implies that $q\in W^{\icomp}$ (because the condition $\ell(w)=\ell(y)+\ell(q)$ implies $D_R(q)\cap \icomp\subseteq D_R(w)\cap \icomp=\varnothing$). Lemma \ref{lem:B} now implies that
\[
r_0(wW_{\icomp})=\beta_I(w)\,W_{\icomp}=
\beta_I(yq)\,W_{\icomp}\in \{\beta_I(y)\,W_{\icomp}, \beta_I(y)s\,W_{\icomp}\}\subseteq \rho_s^{-1}(\beta_I(y)W_J),
\]
so $r_0$ maps every coset in the fiber $\rho_s^{-1}(yW_J)$ to the fiber $\rho_s^{-1}(\beta_I(y)W_J)$. Consequently, we have $\rho_s\circ r_0=r_0^{(s)}\circ\rho_s$
for the map $r_0^{(s)}: W/W_{J}\rightarrow W/W_J, yW_J\mapsto \beta_I(y)W_J$.

{(c)} Suppose $r_0(vW_{\icomp})=C_0$ with $v\in W^{\icomp}$. Then
$\beta_I(v)\in W_{\icomp}\cap W_I=\{1_W\}$, so $\beta_I(v)=1_W$. If $v\ne 1_W$, fix a reduced word $\bfs$ for $v$. If all letters in $\bfs$ are in 
${\icomp}$, then $v\in W_{\icomp}\cap W^{\icomp}=\{1_W\}$, a contradiction; if some letter $s$ in $\bfs$ is from $I$,  then \eqref{cond: demazure} implies that $s\le\beta_I(v)=e$, which is again a contradiction. Therefore $v=1_W$ and
$r_0^{-1}(C_0)=\{C_0\}$.

We have proved that $r_x=r_0$ satisfies conditions (a)--(c) when $x=e$. 
Let us prove the case of a general $x$. We have $x^{-1}\Sigma_x=\Sigma_0$ and $x^{-1}C=C_0$, and $\rho_s$ is $W$-equivariant, so
(a)--(c) transfer verbatim, with $r^{(s)}(Q):=x\,r_0^{(s)}(x^{-1}Q)$ for every coset $Q$ of $W_{\icomp\cup\{s\}}$.
\end{proof}

We will refer to each map of the form $r_x$ as a \emph{rank-selected retraction}. 
These retraction maps form the core component of the following proof, as well as of the decoder of Coxeter codes 
in Sec.~\ref{sec:decoding}; see the example in Sec.~\ref{sec: Example} for an illustration.

\begin{proposition}[Shadow distance]\label{prop:14}
Let $I\subseteq S$. The distance of the shadow code equals
\[
\dist(D_I)=|W_I| .
\]
\end{proposition}

\begin{proof} 
If $I=\varnothing$, then ${\icomp}=S$, and
$D_I=D_\varnothing=\FF_2[W/W_S]\cong\FF_2$ has distance $1=|W_I|$ as desired, so assume $I\ne\varnothing$.

The element $\pi_{\icomp}(b_I)\in D_I$ has weight at most $\wt(b_I)=\abs{W_I}$ by Proposition \ref{prop:pi}(b), so we have $\dist(D_I)\le|W_I|$ and it remains to prove the reverse inequality. Let $z\in D_I\backslash\{0\}$. 
Choose a coset $C=xW_{\icomp}\in\supp(z)$ with minimal representative $x\in W^{\icomp}$. This is a chamber in the $I$-apartment $\Sigma:=\Sigma_x$.
Now let
$r_x:W/W_{\icomp}\to\Sigma_x$ be as in Lemma~\ref{lem:C}. Overloading the notation, let us extend $r_x$ linearly to
a map $r_x: M_{\icomp}\to\FF_2[\Sigma_x]\subseteq M_{\icomp}$, and set $z_\Sigma:=r_x(z)$. 
Note that $z_\Sigma\ne0$: by property (c) of Lemma~\ref{lem:C}, $C$ is the only chamber in $\Sigma_x$ mapped to $C$ by $r_x$, so the
coefficient of $C$ in $z_\Sigma$ equals its coefficient in $z$, namely, $1$.

For any $s\in I$, we have
$\rho_s(z)=0$ by Lemma~\ref{lem:shadow}, so it follows from 
 property (b) of Lemma~\ref{lem:C} that
\[
\rho_s(z_\Sigma)=\rho_s\,r_x(z)=r^{(s)}\,\rho_s(z)=0;
\]
that is, $z_\Sigma$ vanishes under every $s$-coarsening map.

Now write $z_\Sigma=\sum_{u\in W_I}a_u\,xu\,W_{\icomp}$ with $a_u\in\FF_2$. Fix $s\in I$.
Two chambers $xuW_{\icomp},\,xvW_{\icomp}$ of $\Sigma_x$ lie in the same $\rho_s$-fiber if and only if 
$xuW_{{\icomp}\cup\{s\}}=xvW_{{\icomp}\cup\{s\}}$, i.e.,\ if and only if
\[
u^{-1}v\in W_I\cap W_{{\icomp}\cup\{s\}}=W_{\{s\}}=\{e,s\}.
\]
This implies that the
$\rho_s$-fibers in $\Sigma_x$ are exactly the pairs $\{xuW_{\icomp},\,xus\,W_{\icomp}\}$ where $u\in W_I$. The fact that
$\rho_s(z_\Sigma)=0$ now forces $a_u+a_{us}=0$, i.e., 
\[
 a_u=a_{us},\quad\text{for all } u\in W_I,\ s\in I.
\]
These relations say that $u\mapsto a_u$ is constant along the edges of the Cayley
graph of $W_I$ with generating set $I$. Since this graph is
connected, all the coefficients $a_u$ are equal; since $z_\Sigma\ne0$, they all equal $1$. It follows that
$z_\Sigma=\sum_{u\in W_I}xu\,W_{\icomp}$ and $\wt(z_\Sigma)=|W_I|$.

Finally, the linear map $r_x$ sends each basis element of $M_{\icomp}$ to a basis element of $\FF_2[\Sigma_x]$, so it can only merge
or cancel coordinates and never increases weight, and therefore
$\wt(z)\ge\wt(z_\Sigma)=|W_I|$. As our choice of $z\in D_I\backslash\{0\}$ was arbitrary, it follows that $\dist(D_I)\ge|W_I|$, as desired.
\end{proof}

\begin{proof}[Proof of Theorem \ref{thm: main}]
We prove this by induction on $r$. For the base case, the code $\Coxeter{0}=\{0^{|W|},1^{|W|}\}$ has distance $|W|=d_0$.
Now let us assume that the statement is true for $\Coxeter{r-1}$ for some $r\ge 1$. Let $c\in \Coxeter{r}\backslash
\{0\}$. If $\phi_r(c)=0$, then by Theorem~\ref{thm:ses} we conclude that $c\in \Coxeter{r-1}$, and thus
$\wt(c)\ge d_{r-1}$ by the induction hypothesis. 
Next, notice that $d_{r-1}\ge d_r$. Indeed, let $I\subseteq S, |I|=m-r+1$ be such that $d_{r-1}=|W_I|$.
Take $J\subset I$ with $|J|=m-r$ and observe that $d_{r-1}= |W_I|\ge |W_J|\ge d_{r}$. In conclusion, $\wt(c)\ge d_r$. 

Now suppose that $\phi_r(c)\ne 0$, then for some $I$ of size $m-r$, we have $\pi_{I^c}(c)\in D_I\backslash\{ 0\}$, and so
$\wt(\pi_{I^c}(c))\ge |W_I|$ by Proposition~\ref{prop:14}. Now Proposition~\ref{prop:pi}(b) implies that $\wt(c)\ge \wt(\pi_{I^c}(c))\ge d_r$, completing the induction. Finally, note that if $d_r$ is attained for $W_I$ with some
$I$ of size $m-r$, then the code $\Coxeter{r}$ contains the indicator vector of $b_I$, proving the equality in the statement.
\end{proof}

\begin{remark}
We note that in the case where $r\ge \lfloor \frac m2\rfloor$, the conclusion of Theorem 
\ref{thm: main} is already proved in \cite[Cor.4.7]{coble2026coxeter}, where the proof 
uses the well-known classification of finite Coxeter systems and the fact that all such systems have bipartite Coxeter graphs. Our approach via the reduction to shadow codes is somewhat more involved by comparison, but it allows us to treat all possible values of $r$ at once, without relying on the classification.
\end{remark}


\section{Majority-logic decoding of Coxeter codes}\label{sec:decoding}
Shortly after the discovery of RM codes by Muller \cite{muller1954application}, Reed introduced an algorithm for their decoding that corrects any combination of errors up to half the code's minimum distance \cite{reed1954codes}, which we now generalize to arbitrary Coxeter codes. To build intuition, we begin with a brief overview of Reed's decoder.
The encoding map of $RM(r,m)$ sends an {\em information vector} $\mu\in\ZZ_2^{\binom{m}{\le r}}$ to a codeword $c\in \ZZ_2^{2^m}$, where $\binom{m}{\le r}:=\sum_{i=0}^r\binom mi$ is the code's dimension. Specifically, consider the set $V_r=\{v\in \ZZ_2^m:0\le\wt(v)\le r\}$ with $|V_r|=\binom m{\le r}$, and write $\mu=(\mu_v)_{v\in V_r}$.  Define a Boolean polynomial $f(x_1,\dots,x_m)=\sum_{v\in V_r} \mu_v x_1^{v_1}\dots x_m^{v_m}$; then $\mu$ is encoded into the code vector $c=(c_w, w\in \ZZ^{2^m})$ such that
    $$
   c_w=f(w_1,\dots,w_m) \qquad\text{for all }w=(w_1,\dots, w_m).
   $$
  
Reed's algorithm recursively recovers the coefficients $\mu_v$ for $v$ of decreasing Hamming weight and is 
a combination of the following observations:
\begin{enumerate}
    \item The code $RM(r,m)$ is spanned by indicators of all affine subspaces $\langle e_{i_1},\dots,e_{i_k}\rangle$ of $\ZZ_2^m$ of dimension $k\ge (m-r)$;
  
    \item Given an $(m-r)$-dimensional (linear) subspace $L\subset \ZZ_2^m$ spanned by $e_1,\dots,e_{m-r}$ (say) the subspace spanned by $e_{m-r+1},\dots, e_m$ and its cosets each intersect $L$ on a single symbol. These
    intersections form parity checks for the coefficients $\mu_v$ with $\wt(v)=r$. If the count of errors is not too high,
    the majority vote recovers all $\mu_v$ correctly. Subtracting the corresponding part of $f(x)$, we reduce the decoding 
    problem to the code $RM(r-1,m)$, whose distance is twice that of $RM(r,m)$, and repeat the majority vote for the coefficients $\mu_v, \wt(v)=r-1$. {At the last step of the recursion, $r=0$,  the remaining part of $f(x)$ is just the constant term $\mu_0$, and the corresponding code is the repetition code $RM(0,m)$.}
\end{enumerate}
See \cite[Sec.13.6]{MS77} for a detailed presentation. 

In this section, we extend this idea to Coxeter codes. While this requires some adjustments, the general approach still relies on majority votes, and it specializes to the original RM decoder sketched above if $W=\ZZ_2^m$.

Fix $r\in\{0,1,\dots,m\}$. By Proposition~\ref{prop:basis},  the code $\Coxeter{r}$ has the descent basis
\[
  \mathcal{C}_{m-r}=\{\,e_w : |D_R(w)|\ge m-r\},
  \qquad
  e_w:=w\,b_{D_R(w)}=\1_{wW_{D_R(w)}}.
\]
Thus a codeword $c\in \Coxeter{r}$ is the
encoding of its \emph{information symbols} $(\mu_w)_{|D_R(w)|\ge m-r}$, $\mu_w\in\FF_2$
through
   \begin{equation}\label{eq:codeword}
  c=\sum_{w:\,|D_R(w)|\ge m-r}\mu_w\,e_w .
  \end{equation}
Recall that the shadow map $\phi_k$ defined in \eqref{eq: phi map} 
isolates a single descent layer,
\begin{equation}\label{eq:isolate}
  \pi_{\icomp}(c)=\sum_{w:\,D_R(w)=I}\mu_w\,\pi_{\icomp}(wb_I)\ \in\ D_I .
\end{equation}

\subsection{Parity checks from apartments}\label{ss:votes}
Let $y=c+\epsilon\in \FF_2^{|W|}$, where $c\in \Coxeter{r}$ and $\epsilon$ is an error vector. We construct a set of equations (parity checks) that recover the coefficients $\mu_w$ by a majority vote.
For $k\le r,$ fix $I\in\mathcal{I}_{m-k}$ and $w$ with $D_R(w)=I$. Then
$w\in W^{\icomp}$ is the minimal-length representative of the chamber $C_w:=wW_{\icomp}$, and
$\Sigma_w=\{wu\,W_\icomp:u\in W_I\}$ is the apartment `centered' at $C_w$. Let
$r_w:W/W_\icomp\to\Sigma_w$ be the rank-selected retraction of Lemma~\ref{lem:C} centered at $C_w$, with $r_w(\Omega)=w\,r_0(w^{-1}\Omega)$ where $r_0(vW_\icomp)=\beta_I(v)W_\icomp$ for any minimal
representative $v\in W^{\icomp}$. 
Its fibers partition $W/W_\icomp$, so pulling back along
$W\to W/W_\icomp$ partitions $W$ into $|W_I|$ blocks indexed by $u\in W_I$:
\begin{equation}\label{eq:blocks}
  T_w^{(u)}:=\bigl\{\,g\in W:\ r_w(gW_\icomp)=wu\,W_\icomp\,\bigr\}
           =\bigl\{\,g\in W:\ \beta_I\!\bigl((w^{-1}g)^{\icomp}\bigr)=u\,\bigr\},
\end{equation}
where $(\,\cdot\,)^{\icomp}$ denotes the minimal-length representative modulo $W_\icomp$; see \eqref{cond: factor}. The \emph{$u$-th parity check at $w$} of $y$ is
\[
  V_w^{(u)}(y):=\sum_{g\in T_w^{(u)}}y_g\ \in\ \FF_2,\qquad u\in W_I .
\]
The following lemma is formulated to fit the recursion of the decoding
algorithm of Sec.~\ref{sec: decoder} below; in particular, $k$ refers to the order
of the subcode $\Coxeter{k}\subset \Coxeter{r}$ that arises in the corresponding recursion step.
\begin{lemma}\label{lem: votes} \phantom{Lemma + corollary}\\
{\rm(a)} For $k\in\{0,1,\dots,r\}$,
let $I\in\mathcal{I}_{m-k}$, let $w$ satisfy $D_R(w)=I$, and let
$c'=\sum_{w'}\mu_{w'}e_{w'}\in \Coxeter{k}$ be a codeword with
\begin{equation}\label{eq:hyp}
  \mu_{w'}=0\quad\text{for every }w'\text{ with }D_R(w')=I\text{ and }w'>w .
\end{equation}
Then $V_w^{(u)}(c')=\mu_w$ for every $u\in W_I$. 

\noindent{\rm(b)} Let $\epsilon\in \FF_2^{|W|}$ and let $y=c'+\epsilon$.
If $\wt(\epsilon)<\tfrac12|W_I|$, then $\text{\rm maj}_{u\in W_I}V_w^{(u)}(y)=\mu_w$.
\end{lemma}

\begin{proof} (a)
Each block $T_w^{(u)}$ is a union of left $W_\icomp$-cosets, so $V_w^{(u)}(y)$ equals the
sum, over the fiber $r_w^{-1}(wuW_\icomp)$, of the coordinates of $\pi_\icomp(y)$; that is,
$V_w^{(u)}(y)$ is the coefficient of $wuW_\icomp$ in the vector
$r_w\bigl(\pi_\icomp(y)\bigr)\in\FF_2[\Sigma_w]$.

We will first assume that $\epsilon=0$ and so $y=c'$, and discuss the general $\epsilon$ later. By \eqref{eq:isolate} we have $z:=\pi_\icomp (c')\in D_I$, so $\rho_s(z)=0$ for all $s\in I$ by
Lemma~\ref{lem:shadow}. The argument of
Proposition~\ref{prop:14} then shows that $r_w(z)$ has all its coefficients on $\Sigma_w$ equal,
with common value the coefficient of $z$ at the center $C_w$ (the center has the unique
preimage $r_w^{-1}(C_w)=\{C_w\}$). Hence
\[
  V_w^{(u)}(c')=[\,\pi_\icomp (c')\,]_{C_w}\quad\text{for all }u\in W_I .
\]
Now $\pi_\icomp(w'b_I)=\sum_{u'\in W_I}w'u'W_\icomp$ is the indicator of the chamber set $\Sigma_{w'}$ with every coefficient equal to 1 since the terms in the sum are pairwise distinct basis vectors of $M_{I^c}$. Therefore, its coefficient at $C_w$ is
$N(w,w'):=\1[\,C_w\in\Sigma_{w'}\,]\in\{0,1\}$, and by \eqref{eq:isolate}
\begin{equation}\label{eq: unitriangular}
  [\,\pi_\icomp(c')\,]_{C_w}=\sum_{w':\,D_R(w')=I}\mu_{w'}\,N(w,w') .
\end{equation}
Suppose that $N(w,w')=1$, i.e., $C_w\in\Sigma_{w'}$, or $wW_\icomp\cap w'W_I\neq\varnothing$.
Choose $a\in W_I$ with $w'a\in wW_\icomp$. As $D_R(w')\supseteq I$, the element $w'$ is the
longest and hence the Bruhat-largest element of $w'W_I$ by \eqref{cond: coset_rep}, so $w'a\le w'$;
and $w=(w'a)^{\icomp}\le w'a$ since a minimal coset representative lies below every
element of its coset. 
Thus $N(w,w')=1$ implies that $w\le w'$, and $N(w,w)=1$ since $C_w\in \Sigma_w$, so the matrix $N$ is triangular with 1's on the main diagonal. In the sum in \eqref{eq: unitriangular} every term with $w'\ne w$ vanishes. Specifically, the terms with $w'< w$ vanish because $N(w,w')=0$, and
the terms with $w'>w$ vanish because $\mu_{w'}=0$ by \eqref{eq:hyp}. Only the diagonal term $N(w,w)\mu_w$ remains,
so $[\,\pi_\icomp(c')\,]_{C_w}=\mu_w$, and therefore $V_w^{(u)}(c')=\mu_w$ for all $u$,
proving Part~(a).

(b) For $y=c'+\epsilon$ we have
$V_w^{(u)}(y)=\mu_w+\sum_{g\in T_w^{(u)}}e_g$. As the blocks $T_w^{(u)}$ are pairwise
disjoint, the number of $u$ with $\sum_{g\in T_w^{(u)}}e_g\neq0$ is at most $\wt(\epsilon)$; so if $\wt(\epsilon)<\tfrac12|W_I|$, then fewer than half of the $|W_I|$ votes are flipped, and
the majority returns $\mu_w$.
\end{proof}

\begin{remark}\label{remark: triangular} Assumption \eqref{eq:hyp} is added to simplify the processing; without it,
we would still be able to arrive at the same conclusions, but the assembled votes would form a
triangular system of equations, requiring an extra step for the recovery of the individual coefficients. 
\end{remark}

\subsection{The decoder}\label{sec: decoder}
\begin{quote}
\noindent\textbf{Majority-logic decoder for $\Coxeter{r}$.}\\
\emph{Input}: a vector $y=c+\epsilon\in \FF_2^{|W|}$\\
 \emph{Output}: symbols
$(\mu_w)_{|D_R(w)|\ge m-r}$, cf. Eq.~\eqref{eq:codeword}
\begin{enumerate}
\item For $k=r,\,r-1,\,\dots,\,0$:
  \begin{enumerate}
  \item For each $I\in\mathcal{I}_{m-k}$ and each $w$ with $D_R(w)=I$, processed in
        order of \emph{decreasing} length $\ell(w)$:
    \begin{enumerate}
    \item compute the votes $V_w^{(u)}(y)$ for $u\in W_I$ via \eqref{eq:blocks}, and set
          \[
            \mu_w:=\text{maj}_{u\in W_I}\,V_w^{(u)}(y);
          \]
    \item if $\mu_w=1$, update $y\leftarrow y+e_w$.
    \end{enumerate}
  \end{enumerate}
\item Return $(\mu_w)_w$.
\end{enumerate}
\end{quote}

For $k=0$ one has $\mathcal{I}_m=\{S\}$, $W_S=W$ and ${\icomp}=\varnothing$; the $\icomp$-shadow projection is the identity. 
In this case, the only element $w$ with $D_R(w)=I$ is the longest element $w=w_0$, and there is a single $\icomp$-apartment, where each chamber (a coset of $W/\{e\}$) corresponds to a single group element. 
Step (1)(a)(i) of the decoder degenerates to a global majority of the coordinates of
$y$, so Step (1)(a)(ii) returns $\mu_{w_0}$, consistently with $\Coxeter{0}=\{0,\1_W\}$ and $e_{w_0}=\1_W$. This case
is fully analogous to the decoder of RM codes described above.

In accordance with Remark~\ref{remark: triangular}, decoding the symbols of a fixed type $I$ in decreasing length removes, before $w$ is
treated, every element $w'>w$ of the same type; this guarantees that the hypothesis \eqref{eq:hyp} holds
and lets the vote read off $\mu_w$ directly.  Peeling the layer $\mathcal{I}_{m-k}$ before $\mathcal{I}_{m-k+1}$ is exactly the
passage from $\Coxeter{k}$ to $\ker\phi_k=\Coxeter{k-1}$ in Theorem~\ref{thm:ses}.

\begin{theorem}\label{thm:decoder}
Let $c\in \Coxeter{r}$ and $y=c+\epsilon$ with $\wt(\epsilon)\le d_r/2-1$. Then the
majority-logic decoder returns the information symbols $\mu_w$ of $c$; see Eq.~\eqref{eq:codeword}. \end{theorem}
\begin{proof}
Put $t:=\wt(\epsilon)\le d_r/2-1$, so $2t<d_r$. We show by downward induction
on $k$ that, when the outer loop reaches index $k$, the current word equals
$y=c_k+\epsilon$, where $c_k:=\sum_{w:\,|D_R(w)|\ge m-k}\mu_w e_w$ is the part of $c$ of descent
number $\ge m-k$, and that all symbols $\mu_w$ with $|D_R(w)|=m-k$ are then computed
correctly.

At step $k$, the residual is $y=c_k+\epsilon$ (for $k=r$ this is the input
$y=c+\epsilon$). Fix $I\in\mathcal{I}_{m-k}$ and process the elements $w$ with $D_R(w)=I$ in the order of
decreasing length. Suppose that the currently processed element is $w$, then the current residual word is
$y=c'+\epsilon$, where $c'$ is obtained from $c_k$ by deleting every already processed symbol as in Step (1)(a)(ii). The deleted symbols
are precisely the same-type elements whose lengths exceed $\ell(w)$; in particular, every $w'$ with
$D_R(w')=I$ and $w'>w$ (and thus also $\ell(w')>\ell(w)$) has been deleted, so $c'$
satisfies \eqref{eq:hyp}. Moreover, $c'\in \Coxeter{k}$, so Lemma~\ref{lem: votes} implies that
the vote value in the absence of errors is $\mu_w$. Since 
   \begin{equation}\label{eq: t}
   |W_I|\ge d_k\ge d_r>2t\ge 2\wt(\epsilon), 
   \end{equation}
the majority returns $\mu_w$ correctly. Step (1)(a)(ii) then subtracts $\mu_w e_w$, so the structure is preserved for the next step.

Once all $w$ with $|D_R(w)|=m-k$ are processed, the residual vector is $c_{k-1}+\epsilon$, where the error $\epsilon$ is unchanged from $y$, so the recursion can continue.
At $k=0$ the single symbol $\mu_{w_0}$ is recovered by the global majority,
again correct as above. After the loop, the residual vector is $\epsilon$ and all
information symbols have been recovered.
\end{proof}

Theorem \ref{thm:decoder} implies that errors of multiplicity $t$ are corrected as long as $t$ satisfies Eq.\eqref{eq: t}. Since $d_r$ is even, this shows that $\dist(\Coxeter{r})\ge d_r-1$, stopping one short of the exact value. The loss occurs because of the ties that can arise in the case of exactly $d_r/2$ errors; the combinatorial proof of Theorem~\ref{thm: main} corresponds to error {\em detection} rather than correction and thus avoids this issue.
Phrased differently, if the decoding algorithm is used to detect errors and returns a codeword only if all the votes agree, then to fail it needs $\wt(\epsilon)\ge d_r$, showing again that $\dist(\Coxeter{r})=d_r$.

\begin{remark}(The case of RM codes) The opening paragraph of this section contains an informal description of Reed's decoding of RM codes. Relying on the above discussion, we present it concisely and more formally.
When $W=\ZZ_2^m$ (type $mA_1$), $W_I$ is the coordinate subgroup on $I$, a chamber
$wW_{\icomp}$ is the $r$-flat through $w$ parallel to $\langle I\rangle$, and the
apartment $\Sigma_w$ is the partition of $\ZZ_2^m$ into the cosets of that flat. Since the
flats partition the group,
the retraction map is trivial, the blocks $T_w^{(u)}$ coincide with chambers, the matrix
$N(w,w')$ of Lemma~\ref{lem: votes} is the identity, and the decreasing-length ordering is
vacuous. The votes recover all degree-$r$ coefficients in parallel, and this constitutes one step of the original Reed decoding procedure.   
\end{remark}

\subsection{Example}\label{sec: Example} Let $W=S_4, S=\{s_i=(i,i+1): 1\le i\le 3\}$. 
Consider a decoding step of the code $\coxeter{S_4}{}{1}$ with the parameters $[n=24, \dim=12, \dist=4]$.
We will illustrate the forming of parities for one coefficient $\mu_w$ in \eqref{eq:codeword}, taking  $w=s_1s_3=2143$ (in one-line notation). We note that the operations described below do not depend on $c$ or on $\epsilon$: the processing
is exactly the same, and if the error weight exceeds the correction radius, the vote may return an incorrect value.

\begin{table}[ht]
{\small \begin{tabular}{c|cccccc}
 \diagbox{chambers}{apartments}& $\{1,2\}$ & $\{1,3\}$ & $\{1,4\}$ & $\{2,3\}$ & $\{2,4\}$ & $\{3,4\}$\\
\hline
$(1,2)$ & $\cdot$        & $1342$ & $1432$ & $\cdot$ & $\cdot$ & $\cdot$\\
$(1,3)$ & $1243$         & $\cdot$ & $1423$ & $\cdot$ & $\cdot$ & $\cdot$\\
$(1,4)$ & $1234$         & $1324$ & $\cdot$ & $\cdot$ & $\cdot$ & $\cdot$\\
$(2,1)$ & $\cdot$        & $\cdot$ & $\cdot$ & $2341$ & $2431$ & $\cdot$\\
$(2,3)$ & $\mathbf{2143}$& $\cdot$ & $\cdot$ & $\cdot$ & $2413$ & $\cdot$\\
$(2,4)$ & $2134$         & $\cdot$ & $\cdot$ & $2314$ & $\cdot$ & $\cdot$\\
$(3,1)$ & $\cdot$        & $\cdot$ & $\cdot$ & $3241$ & $\cdot$ & $3421$\\
$(3,2)$ & $\cdot$        & $3142$ & $\cdot$ & $\cdot$ & $\cdot$ & $3412$\\
$(3,4)$ & $\cdot$        & $3124$ & $\cdot$ & $3214$ & $\cdot$ & $\cdot$\\
$(4,1)$ & $\cdot$        & $\cdot$ & $\cdot$ & $\cdot$ & $4231$ & $4321$\\
$(4,2)$ & $\cdot$        & $\cdot$ & $4132$ & $\cdot$ & $\cdot$ & $4312$\\
$(4,3)$ & $\cdot$        & $\cdot$ & $4123$ & $\cdot$ & $4213$ & $\cdot$\\
\end{tabular}}
\caption{Partition of $W$ for the decoding of $\mu_w,w=2143$.}
\label{table: grid}
\end{table}

Table~\ref{table: grid} shows the partition of $W$ for the recovery of $\mu_w$. It depends on $w$ through
$I=D_R(w)=\{s_1,s_3\}$. We have $W_I=\langle s_1\rangle\times \langle s_3\rangle$ and $W_{\icomp}=\langle s_2\rangle$, so every $I$-chamber $uW_{\icomp}$ contains two elements determined by the pair $(p,q):=(u(1),u(4))$, and every $I$-apartment is determined by four elements $x$ that form the set $\{x(1),x(2)\}$. The rows of Table~\ref{table: grid} show the partition of $W$ into 12 chambers, with the pairs $(p,q)$ as the row labels, and each column lists the four elements $x$ indexing the same apartment, with the shared set $\{x(1), x(2)\}$ as the corresponding label.

The retraction mapping of Lemma~\ref{lem:C} provides a way to assemble the votes for $\mu_w$ from the coordinates of $y$ into the blocks $T_w^{(u)}$ as in Eq.~\eqref{eq:blocks}. 
Recall that the retraction $r_w$ centered at $w$ is defined as $r_w(\Omega)=w\,r_0(w^{-1}\Omega)$  with $r_0(vW_\icomp)=\beta_I(v)W_\icomp$ for the minimal
representatives $v\in W^{\icomp}$. 
In our case, $w=w^{-1}=s_1s_3=2143=(12)(34)$ (in the cycle notation), and our labelling scheme for the chambers guarantees that  $w(p,q)=(w(p),w(q))$. Furthermore, if $v\in W^{\icomp}$, $vW_\icomp=(p,q)$, and $\beta_I(v)W_\icomp=(p',q')$ (in other words, if $v$ and $\beta_I(v)$ send $(1,4)$ to $(p,q)$ and $(p',q')$, respectively, in the coordinate-wise action), then we have the following:
\begin{itemize}
\item[(a)] If $p=v(1)=1$, then $v$ can be generated by $S\setminus \{s_1\}$, so its reduced words do not include $s_1$ and the same is true for $\beta_I(v)$, and therefore $p'=1$.
\item[(b)] If $p=v(1)>1$, then $v$ cannot be generated by $S\setminus\{s_1\}$, so
every reduced word of $v$ includes $s_1$; equivalently, $s_1\le v$ in the Bruhat
order. Since $s_1\in W_I$, the maximality of $\beta_I(v)$ in $\{u\in W_I:u\le v\}$
(Lemma~\ref{lem:B}) gives $s_1\le\beta_I(v)$, so $\beta_I(v)\in\{s_1,s_1s_3\}$;
in either case, $p'=2$.
\end{itemize}
Similarly, we have $q'=4$ if $q=u(4)=4$ and $q'=3$ otherwise, so we have 
 \begin{equation}
  \label{eq:pq1}
   r_0((p,q))=(p', q'), \quad \text{where $p'=1$ if $p=1$, else 2; and \text{$q'= 4$ if $q=4$, else 3}}.
\end{equation}
It further follows that
  \begin{equation}
  \label{eq:pq}
   r_w((p,q))=(p'', q''), \quad \text{where $p''=2$ if $p=2$, else 1; and \text{$q''= 3$ if $q=3$, else 4}}.
\end{equation}
For example, if $(p,q)=(2,1)$, then $r_w((p,q))=(2,4)$.

We can use Eq. \eqref{eq:pq} to calculate the fibers of the retraction $r_w$. The outcome of these calculations is summarized in Table~\ref{table: votes}. Note that the ``center'' chamber $C_w:=wW_{\icomp}$ has preimage $r_w^{-1}(C_w)=\{C_w\}$, as promised by Lemma~\ref{lem:C}(c).

\begin{table}[ht]{\small
\begin{tabular}{|c l p{3in}|}\hline
anchor & block (chambers) & group elements ($\supp(e_w)$ in bold)\\
\hline
$(2,3)$ & $(2,3)$ & $\mathbf{2143},\ 2413$ \\
$(2,4)$ & $(2,4),(2,1)$ & $\mathbf{2134},\ 2314,\ 2341,\ 2431$ \\
$(1,3)$ & $(1,3),(4,3)$ & $\mathbf{1243},\ 1423,\ 4123,\ 4213$ \\
$(1,4)$ & $(1,4),(1,2),(3,1),$  & $\mathbf{1234},\ 1324,\ 1342,\ 1432,\ 3241,\ 3421,\ 3142$, \\
&$(3,2),(3,4),(4,1),(4,2)$ &\qquad$\ 3412,\ 3124,\ 3214,\ 4231,\ 4321,\ 4132,\ 4312$\\
\hline
\end{tabular}
}
\caption{The votes table for the coefficient $\mu_{w=(2143)}$. The elements in the support of $e_w$, i.e., the elements of $w W_I$, are typeset in boldface. Anchors refer to the chambers that contain those elements.}\label{table: votes}
\end{table}

\begin{remark}\label{remark: noncommuting}
For better readability, we chose $W_I=\langle s_1,s_3\rangle$ with commuting generators, resulting in $|W_I|=4$. For the case of non-commuting generators, we could take $I=\{s_1,s_2\}$, then apartments become hexagons rather than squares, every pair of apartments share 2 chambers (apartments glue along edges), the retraction does not factor because $\beta_I$ is a genuine Demazure product on $\{s_1,s_2\}$-words, and
the forming of the blocks $T_w^{(u)}$ depends on the choice of the direction ($s_1$ or $s_2$).
\end{remark}


\section{Concluding remarks}
General Coxeter codes are defined in combinatorial-geometric terms as in Eq.~\eqref{eq:codeword}. RM codes, in addition, can be 
described relying on polynomial formalism as discussed at the beginning of Sec.~\ref{sec:decoding}. It is not clear to us whether
the polynomial description affords a well-formed extension to the case of general (finite) Coxeter groups.

In \cite[App.~B]{coble2025thesis}, Coble defines the so-called subapartment codes and residue codes, which are two kinds of building-theoretic extensions of Coxeter codes. By definition, a subapartment code is the span of the indicator vectors of subapartments of rank $m-r$ in a finite spherical building of rank $m$, whereas residue codes are spanned by the indicator vectors of residues of rank $m-r$. When the building is the Coxeter complex of $(W,S)$, these two definitions recover the Coxeter code $C_W(r)$. 
This suggests several natural questions beyond the finite Coxeter-group setting. The most immediate questions concern the basic coding-theoretic parameters: determining the dimension and minimum distance of these building codes. It would also be interesting to understand whether the shadow-code and rank-selected retraction methods used here have analogues and whether they can be used to prove dimension or distance statements in this more general setting.

\section{Acknowledgments}
The authors are grateful to Myna Vajha for suggesting exploring links between a 
majority-logic decoder of Coxeter codes and the minimum distance conjecture, and to James Davis for helpful discussions.
A.B. was supported in part by NSF grants CCF-2330909 and CCF-2526035; T.X. was supported in part by an AMS–Simons Research Enhancement Grant for PUI Faculty.
This project was initiated while the authors took part in the workshop ``The Interplay Between Distance Geometry, Combinatorics, and Coding Theory" organized at the Brin Mathematics Research Center at the University of Maryland, College Park, in November 2025. 
The authors used large language models, including ChatGPT-5.4 and Claude Sonnet 4.6, as exploratory tools in connection with the problem studied in this paper. Their use was limited to brainstorming, discussion of possible approaches, and preliminary checking of ideas. All mathematical arguments were independently written by the authors, who take full responsibility for the content of the paper.

\end{document}